\documentclass[3p, sort&compress]{elsarticle}
\usepackage{amsmath}
\usepackage{color}
\usepackage{amsthm}
\usepackage{amssymb}
\usepackage{titlesec}
\usepackage[thinlines,thiklines]{easybmat}

\usepackage{graphicx}
\usepackage{epsfig}
\theoremstyle{definition}\newtheorem{Df}{Definition}
\theoremstyle{plain}\newtheorem{Th}{Theorem}
\theoremstyle{definition}\newtheorem{Rm}{Remark}
\theoremstyle{definition}
\theoremstyle{plain}
\theoremstyle{plain}
\theoremstyle{plain}\newtheorem{Lm}{Lemma}
\theoremstyle{plain}

\begin{document}
\begin{frontmatter}
\title{ \Large{Communication complexity of  promise problems\\ and their applications to finite automata}}

\author{Jozef Gruska$^{1}$}
\author{Daowen Qiu$^{2}$ }
\author{Shenggen Zheng$^{1,}$\corref{one}}

 \cortext[one]{Corresponding author.\\ \indent{\it E-mail addresses:} zhengshenggen@gmail.com (S. Zheng), gruska@fi.muni.cz (J. Gruska),  issqdw@mail.sysu.edu.cn (D. Qiu).}

\address{

  $^{1}$Faculty of Informatics, Masaryk University, Brno 60200, Czech Republic\\

  $^2$Department of Computer Science, Sun Yat-sen University,
Guangzhou 510006,
  China\\
}

\begin{abstract}

{\em Equality} and {\em disjointness}  are two of the most studied problems in {\em communication complexity}. They have been studied for both classical and also quantum communication and for various models and modes of communication.
Buhrman et al. \cite{Buh98} proved that the exact {\em quantum communication complexity} for a promise version of the equality problem is  ${\bf O}(\log {n})$  while the {\em classical deterministic  communication complexity} is $n+1$ for two-way communication, which was the first impressively large (exponential) gap between quantum and classical (deterministic and probabilistic) communication complexity. If an error is tolerated,  both quantum and probabilistic communication complexities for equality are ${\bf O}(\log {n})$. However, even if an  error is tolerated, the gaps between quantum (probabilistic) and deterministic complexity are not larger than quadratic for the disjointness problem \cite{AA03,Bar02,Kla00,Ks92,Raz92,Raz03}.
It is therefore interesting to ask whether there are some promise versions of  the disjointness problem for which bigger gaps can be shown. We give a positive answer to such a question. Namely,
 we  prove that there exists an exponential gap between  quantum (even probabilistic) communication complexity  and classical deterministic communication complexity of some specific versions of the disjointness problem.

Klauck \cite{Kla00} proved, for any language,  that the state complexity  of exact quantum/classical finite automata, which is a general model of one-way  quantum finite automata,  is not less than the state complexity of an equivalent one-way deterministic finite automata (1DFA).
In this paper we show, using a communication complexity result, that situation may be different for some promise problems. Namely, we show for certain promise problem that the gap between the
 state complexity of exact one-way  quantum finite automata and 1DFA can be exponential.

\end{abstract}


\end{frontmatter}

\section{Introduction}

Since the topic of communication complexity was introduced by Yao \cite{Yao79} in 1979, it has been  extensively studied \cite{Bra03,Buh09,Hro97,KusNis97}.
In the setting of two distributed parties, Alice is given as an input $x\in\{0,1\}^n$ and Bob is given as an input $y\in\{0,1\}^n$ and their task is to communicate in order to be able to compute the value of some boolean function $f(x,y)\in\{0,1\}$, while exchanging as small  number of bits between Alice and Bob  as possible. In the general case, local computation
is free, but communication is expensive and has to be minimized.
Alice and Bob each have a Turing machine (probabilistic, quantum) and can use all the computation power of their machines.
There are three kinds of communication complexities according to the models (or protocol) used by Alice and Bob: deterministic, probabilistic or quantum.

The communication between Alice and Bob could be  one-way, two-way or simultaneous.
 In {\em one-way communication},  Alice sends a single message to Bob who then is to determine the output.  In two way communication, Alice and Bob take turns in sending messages to each other. In the simultaneous communication mode, both Alice and Bob  send their messages to a third party (called the referee) who, upon receiving messages from Alice and Bob, determines the output.

Two of the most studied communication problems  are equality and disjointness \cite{KusNis97}, which are defined, given a string $x$ as an Alice's input and a string $y$ as a Bob's input, as follows:
\begin{itemize}
\item {\bf Equality}: $\text{EQ}(x,y)=1$ if $x=y$ and 0 otherwise.
  \item {\bf Disjointness}: $\text{DISJ}(x,y)=1$  if there is no index $i$ such that $x_i=y_i=1$ and $0$ if such an index exists.   Equivalently, we can define also this function as
  $\text{DISJ}(x,y)=1$ if $\sum_{i=1}^n x_i\wedge y_i=0$ and  $0$ if $\sum_{i=1}^n x_i\wedge y_i>0$. (We can think of $x$ and $y$ as being subsets of $\{1,\cdots,n\}$ represented by  characteristic vectors and to have $\text{DISJ}(x,y)=1$ if and only if these two subsets are disjoint, i.e., $x\cap y=\emptyset$.)
\end{itemize}

The deterministic  communication complexities (if we do not point out explicitly in this paper, otherwise communication complexity will refer always to two-way communication complexity) for $\text{EQ}$ and $\text{DISJ}$ problems are both $n+1$ \cite{KusNis97}. Buhrman et al. \cite{Buh98} proved that the exact quantum communication complexity for the following promise version of the equality problem
 \begin{equation}
\text{EQ}'(x,y)=\left\{\begin{array}{ll}
                    1 &\ \text{if}\ H(x,y)=0 \\
                    0 &\ \text{if}\ H(x,y)=\frac{n}{2}
                  \end{array}
 \right.,
\end{equation}
is  ${\bf O}(\log {n})$, where $H(x,y)$ is the Hamming distance between $x$ and $y$, which is the number of bit positions on which they
differ.
This was the first impressively large (exponential) gap between quantum and classical (deterministic and probabilistic) communication complexity. If errors are tolerated,  both quantum and probabilistic communication complexity for equality are ${\bf O}(\log {n})$.

Concerning  communication complexity for the disjointness problem,  even on the case an error is tolerated, the probabilistic communication complexity for disjointness is ${\bf \Omega}(n)$ \cite{Bar02,Ks92,Raz92}.
   In the quantum cases, Buhrman et al.  \cite{Buh98} proved that the quantum communication complexity for $\text{DISJ}$ is ${\bf O}(\sqrt{n}\log n)$. This bound was improved to ${\bf O}(\sqrt{n})$ by Aaronson and Ambainis \cite{AA03}. Finally, Razborov
showed that any bounded-error  quantum  protocol for $\text{DISJ}$  needs to communicate about $\sqrt{n}$ qubits \cite{Raz03}.
  However, Klauck \cite{Kla00} showed that the  one-way quantum communication complexity of $\text{DISJ}$  is ${\bf \Omega}(n)$.
It is unlike the problem $EQ$ for which there is an exponential gap between quantum (probabilistic) communication complexity  and deterministic communication complexity \cite{KusNis97,Buh98,Buh09}.  All the known gaps of $\text{DISJ}$ are not larger than quadratic.
It is therefore interesting to find out whether there are some promise versions of the set disjointness problem such that bigger communication complexity gaps can be achieved. In order to prove that,
 we consider the following promise problems in this paper, for $0<\lambda\leq \frac{1}{4}$,
 \begin{equation}
    \text{DISJ}_{\lambda}'(x,y)=\left\{\begin{array}{ll}
                    1 &\ \text{if}\  \sum_{i=1}^n x_i\wedge y_i=0 \\
                    0 &\ \text{if}\ \lambda n\leq \sum_{i=1}^n x_i\wedge y_i\leq (1-\lambda) n\\
                  \end{array}
 \right..
\end{equation}
 We  prove that  quantum communication complexity for $\text{DISJ}_{\lambda}'$ is not more than $\frac{\log 3}{3\lambda}(3+2\log n)$ while the deterministic communication complexity for $\text{DISJ}_{\lambda}'$ is  ${\bf \Omega}(n)$.
  For example, if $\lambda=\frac{1}{4}$, then the quantum communication complexity for $DISJ_{\frac{1}{4}}'$ is not more than $3+2{\log n}$ while the  deterministic communication complexity for $\text{DISJ}_{\lambda}'$ is  more than 0.007n\footnote{The idea of the proof is inspired by \cite{Buh98,Buh09}.}. We prove also that  one-way probabilistic communication complexity for  $\text{DISJ}_{\lambda}'$ is not more than $\frac{\log{3}}{\lambda}\log{n}$.

Number of states is a natural complexity measure for all models of finite automata and  state complexity of finite automata is one of the important research fields with many applications \cite{Yu05}.
There is a variety of methods how to prove lower bounds on state complexity  and methods as well as the results of communication complexity are among the main ones \cite{Hro01,Kla00,KusNis97}. In this paper we also show how to make use of our new communication complexity results to get  new state complexity outcomes.

Klauck \cite{Kla00} has proved that, for any language,  the state complexity of exact  quantum/classical finite automata, which is a general model of one-way  quantum finite automata,   is not less than the state complexity of one-way deterministic finite automata (1DFA).  Therefore, it is interesting and important to find out whether the result still holds for interesting  cases of promise problems or not\footnote{Ambainis and Yakaryilmaz showed in  \cite{AmYa11} that there is a very special case in which the superiority of quantum
computation to classical one cannot be bounded. }.  Applying the communication complexity result from \cite{Buh01,Buh09} to  finite automata,  for any $n\in {\mathbb{Z}}^+$, we prove  that there exists a promise problem $A_{EQ}(n)$ that can be solved exactly by a  {\em measure-once one-way finite automata with quantum and classical state} (MO-1QCFA) with $n$ quantum basis states and ${\bf O}({n})$ classical states, whereas the sizes of the corresponding 1DFA are $2^{{\bf \Omega}(n)}$.
 We then apply  our communication complexity result  to finite automata and prove that, for any $n\in {\mathbb{Z}}^+$, there exists a  promise problem $A_{D}(n)$  that can be solved with one-sided error  by an  MO-1QCFA  with $2n$ quantum basis states and ${\bf O}({n})$ classical states, whereas the sizes of the corresponding 1DFA are $2^{{\bf \Omega}(n)}$.

 The paper is structured as follows. In Section 2 needed basic concepts and notations are introduced and  models (including communication complexity and finite automata)  involved are
described in some details. Communication complexities for $\text{DISJ}_{\lambda}'$ is discussed in Section 3.  Applications to finite automata are studied in Section 4.

\section{Preliminaries}

In this section, we recall some of the basic definitions about communication complexity and quantum finite automata.

\subsection{Communication complexity}

We recall only some very  basic concepts and notations of {\it communication complexity},
and we refer the reader to \cite{Buh09,Kla00-r,KusNis97,Yao79} for
more details. We will deal with the situation that there are  two
communicating  parties and with very simple tasks of
computing two argument functions for the case  one argument is known to
one party and the other argument is known to the other party.
We will completely ignore computational resources needed by
parties and we focus solely on the amount of communication that is need to be
exchanged between both parties in order to compute the value of a function.

More technically, let $X, Y$ be finite sets $\{0,1\}^n$. We consider a two-argument function $f: X\times Y\rightarrow \{0,1\}$ and two communicating parties.  Alice is given an input $x\in X$ and Bob is given an input $y\in Y$. They wish to compute $f(x,y)$.

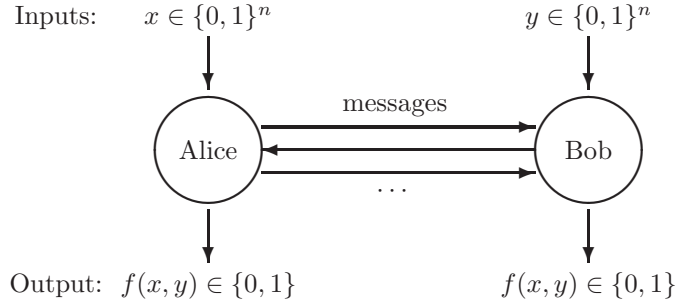
\begin{figure}[htbp]
 %  %Requires \usepackage{graphicx}
  \centering

  \setlength{\unitlength}{1cm}
\begin{picture}(30,4)\thicklines

\put(6,3.5){\vector(0,-1){0.7}\makebox(0,0.5){$x\in\{0,1\}^n$}\makebox(-4,0.5){Inputs:}}
\put(11,3.5){\vector(0,-1){0.7}\makebox(0,0.5){$y\in\{0,1\}^n$}}

\put(11,1.2){\vector(0,-1){0.7}\makebox(0,-2){$f(x,y)\in\{0,1\}$}\makebox(-14,-2){Output:}}

\put(6,1.2){\vector(0,-1){0.7}\makebox(0,-2){$f(x,y)\in\{0,1\}$}}

\put(6,2){\circle{2}\makebox(0,0){Alice}}

\put(11,2){\circle{2}\makebox(0,0){Bob}}

\put(6.7,2.3){\vector(1,0){3.6}\makebox(-3.7,0.5){messages}}
\put(10.3,2){\vector(-1,0){3.6}\makebox(-3.7,0.5)}
\put(6.7,1.7){\vector(1,0){3.6}\makebox(-3.7,-0.5){$\cdots$}}

\end{picture}
  \centering\caption{Communication protocol}\label{f2}
\end{figure}

The computation of   $f(x,y)$ will be  done using a communication protocol, see Figure \ref{f2}. During the
execution of the protocol, the two parties alternate roles in sending messages. Each of these
messages is a bit-string. The protocol, whose steps are based on the communication so far, specifies also for each step whether the communication terminates
(in which case it also specifies what is the output). If the communication is
not to terminate, the protocol specifies what kind of  message the sender (Alice or Bob) should send next, as
a function of its input and communication so far.

A deterministic communication protocol ${\cal P}$ computes
the function $f$, if for every input pair $(x,y)\in A\times B$ the protocol terminates with the
value $f(x,y)$ as its output.
In a probabilistic  protocol, Alice and Bob may flip coins and the protocol can have an erroneous output with a small probability.
In a  quantum protocol, Alice and Bob may   use quantum resources.
 Let ${\cal P}(x,y)$ denote the  output of the protocol ${\cal P}$, we will consider 3 different kinds of protocols for computing $f$:
\begin{itemize}
  \item An exact protocol always outputs the right answer, $Pr({\cal P}(x,y)=f(x,y))=1$.
  \item A one-sided error\footnote{If the error is defined in another side, the communication complexity can be very different in some cases.} protocol that works in such a way that  $Pr({\cal P}(x,y)=f(x,y))=1$ if $f(x,y)=0$ and $Pr({\cal P}(x,y)=f(x,y))\geq \frac{1}{2}$ if $f(x,y)=1$.
  \item A two-sided error (a bounded error) protocol ${\cal P}$ is such that  $Pr({\cal P}(x,y)=f(x,y))\geq \frac{2}{3}$.
\end{itemize}

The communication complexity of a protocol ${\cal P}$  is the
worst case number of (qu)bits exchanged.  The communication complexity of $f$ is, with  which respect to the communication mode used  the complexity of an
optimal protocol for $f$.

In a one-way protocol, shown in Figure \ref{f1}, Alice sends a single message to Bob who then determines the output.

\begin{figure}[htbp]
 %  %Requires \usepackage{graphicx}
  \centering

  \setlength{\unitlength}{1cm}
\begin{picture}(30,4)\thicklines

\put(6,3.5){\vector(0,-1){0.7}\makebox(0,0.5){$x\in\{0,1\}^n$}\makebox(-4,0.5){Inputs:}}
\put(11,3.5){\vector(0,-1){0.7}\makebox(0,0.5){$y\in\{0,1\}^n$}}

\put(11,1.2){\vector(0,-1){0.7}\makebox(0,-2){$f(x,y)\in\{0,1\}$}\makebox(-14,-2){Output:}}

\put(6,2){\circle{2}\makebox(0,0){Alice}}

\put(11,2){\circle{2}\makebox(0,0){Bob}}

\put(6.7,2){\vector(1,0){3.6}\makebox(-3.7,0.5){one message}}

\end{picture}
  \centering\caption{One-way communication protocol}\label{f1}
\end{figure}

We use $D(f)$ to denote  the deterministic complexity of $f$.
We will use $R^1_E(f), R^1_1(f),R^1_2(f)$, $R_E(f), \linebreak[0]R_1(f),R_2(f)$ to denote the one-way and two-way probabilistic communication complexity for $f$ in the exact, one-sided error and two-sided error settings, respectively. Similarly we define $Q^1_E(f), Q^1_1(f),Q^1_2(f)$, $Q_E(f), Q_1(f),Q_2(f)$ for the quantum versions of these communication complexities.

We can summarize some of the previous results about $\text{EQ}$,  $\text{DISJ}$ and promise problem $\text{EQ}'$   as follows:
\begin{enumerate}
  \item $D(\text{EQ})=n+1$, $D(\text{DISJ})=n+1$,  \cite{KusNis97}.
  \item $D(\text{EQ}')\in {\bf \Omega}(n)$, $Q(\text{EQ}')\in {\bf O}(n)$ \cite{Buh98,Buh09}.

  \item $R_2(\text{EQ})\in {\bf O}(\log n)$ \cite{KusNis97}, $R_2(\text{DISJ})\in {\bf \Omega}(n)$ \cite{Bar02,Ks92,Raz92}.
  \item $Q_E(\text{EQ})=Q_1(\text{EQ}) =Q_1(\text{DISJ})= n+1$ \cite{Buh01,HPD02}.
  \item $Q_2(\text{DISJ})\in {\bf \Theta(\sqrt{n})}$ \cite{AA03,Raz03}.
  \item $Q_2^1(\text{DISJ})\in {\bf \Omega}(n)$ \cite{Kla00}.
\end{enumerate}

\subsection{Lower bound methods}

There are quite a lot of lower bound methods  known for classical communication complexity. We just recall the ``rectangles" method  in this paper. Concerning more on lower bound methods, we refer the reader to \cite{Buh09,Hro97,KusNis97}.

A {\em rectangle} of a product of sets $X\times Y$ is a set $R=A\times B$ with $A\subseteq X$ and $ B\subseteq Y$. A rectangle $R=A\times B$ is called
$1(0)$-rectangle of a function $f:X\times Y\to \{0,1\}$
if for every $x\in A$ and $y\in B$ the value of $f(x,y)$ is 1 (0). Moreover $C^i(f)$ is defined  as the minimum number of $i$-rectangles that partition the space of $i$-inputs (such inputs $x$ and $y$ that $f(x,y)=i$).

\begin{Lm}\label{lm-d-lowbound}\cite{KusNis97}
For every function $f:X\times Y\rightarrow\{0,1\}$, $D(f)\geq \log\max\{\log{ C^1(f)},\log{ C^0(f)}\}$.
\end{Lm}
\begin{Rm}
 For a promise problem  $P$ to compute the value of a partial
function $f: X\times Y\to \{0,1\}$, a rectangle $R=A\times B$ is called
$1(0)$-rectangle if  the value of $f(x,y)$ is 1(0) for every $(x,y)\in A\times B$ that is a promise input and we do not care about  values for $(x,y)\in A\times B$ that are not promise inputs.
The above lemma still holds for promise problems (that is for partial functions).
\end{Rm}

\subsection{One-way finite automata with quantum and classical state (1QCFA)}

In this subsection we recall the definition of 1QCFA.  Concerning basic concepts and notations of quantum information processing, we refer the reader to \cite{Gru99,Nie00}, and concerning more on classical and quantum automata \cite{Gru99,Gru00,Hop79,Qiu12}.

{\em Two-way finite automata with quantum and classical states} were introduced by Ambainis and Watrous \cite{Amb02} and explored also by Yakaryilmaz,  Qiu, Zheng and others \cite{Yak10,Zhg12,ZhgQiu11,Zhg13}. Informally, a 2QCFA can be seen as a 2DFA with an access to a quantum memory for states of a fixed Hilbert space upon which at each step either a unitary operation is performed or a projective measurement and the outcomes of which then probabilistically determine the next move of the underlying 2DFA. 1QCFA are one-way versions of 2QCFA \cite{ZhgQiu112}. In this paper, we only use 1QCFA in which a unitary transformation is applied in every step after scanning a symbol and an measurement is performed only after the scanning of the right end-marker. Such model is a measure-once 1QCFA (MO-1QCFA) and corresponds to MO-1QFA, which can also be seen as a special case of {\em one-way quantum finite automata together with classical states} in \cite{Qiu13}.

\begin{Df}
A measure-once 1QCFA ${\cal A}$ is specified by a 10-tuple
\begin{equation}
{\cal A}=(Q,S,\Sigma,\Theta,\Delta,\delta,|q_{0}\rangle,s_{0},S_{acc},S_{rej})
\end{equation}

\begin{enumerate}
\item $Q$ is a finite set of orthonormal quantum (basis) states.
\item $S$ is a finite set of classical states.
\item $\Sigma$ is a finite alphabet of input symbols and let
$\Sigma'=\Sigma\cup \{|\hspace{-1.5mm}c,\$\}$, where $|\hspace{-1.5mm}c$ will be used as the left end-marker and $\$$ as the right end-marker.
\item $|q_0\rangle\in Q$ is the initial quantum state.
\item $s_0$ is the initial classical state.
\item $S_{acc}\subset S$ and $S_{rej}\subset S$, where $S_{acc}\cap S_{rej}=\emptyset$ are  sets of
the classical accepting and rejecting states, respectively.
\item $\Theta$ is a quantum transition function
\begin{equation}
\Theta: S\setminus(S_{acc}\cup S_{rej})\times \Sigma'\to U({\cal H}(Q)),
\end{equation}
where $U({\cal H}(Q))$ is set of unitary operations  on the Hilbert space generated by quantum states from $Q$.

\item $\delta$ is a classical transition function
\begin{equation}
\delta: S\setminus(S_{acc}\cup S_{rej})\times \Sigma'\to S.
\end{equation}
If $\delta(s,\sigma)=(s')$, then the new classical state of the automaton is $s'$.

\item $\Delta$ is the mapping:
\begin{equation}
\Delta:S \to O({\cal H}(Q)),
\end{equation}
where $O({\cal H}(Q))$ is set of projective measurements on the Hilbert space generated by quantum states from $Q$.
\end{enumerate}
\end{Df}
The computation of an MO-1QCFA
${\cal A}=(Q,S,\Sigma,\Theta,\Delta,\delta,|q_{0}\rangle,s_{0},S_{acc},S_{rej})$ on an input $w=\sigma_1 \cdots\sigma_n\in \Sigma^*$ starts with the string $|\hspace{-1.5mm}cx\$$ on the input tape. At the start, the tape head of the automation is positioned on the left end-marker and the automaton begins the computation in the  initial classical state and
in the initial quantum state. After that,
in each  step, if  its  classical state  is $s$, its tape head reads a symbol $\sigma$ and its quantum state is $|\psi\rangle$, then the automaton changes its quantum state to $\Theta(s,\sigma)|\psi\rangle$ and its classical state to $\delta(s,\sigma)$. At the end of the computation, a projective measurement, which has two possible classical outcomes $a$ and $r$,  is applied on the current quantum state. If the classical outcome is $a$ ($r$), then the input is accepted (rejected).

  For any string $w\in (\Sigma')^*$ and any $\sigma\in \Sigma$,
let $\widehat{\delta}(s,\sigma
w)=\widehat{\delta}(\delta(s,\sigma),w)$; if $|w|=0$,
$\widehat{\delta}(s,w)=s$. Let $\sigma_0=|\hspace{-1.5mm}c$ and $\sigma_{n+1}=\$$. Assume that $\widehat{\delta}(s_0,\sigma_0\cdots\sigma_i)=s_{i+1}$.  Suppose the measurement is $M=\{P_a, P_r\}$, then the probability that ${\cal A}$ accepts the input
\begin{equation}
Pr[{\cal A}\ \  \text{accepts}\  w] =\|P_{a}U(s_{n+1},\sigma_{n+1})\cdots U(s_1,\sigma_1)U(s_0,\sigma_0)|q_0\rangle\|^2.
\end{equation}
The probability that ${\cal A}$ rejects the input is $Pr[{\cal A}\ \  \text{rejects}\  w]=1-Pr[{\cal A}\ \  \text{accepts}\  w]$.

Language acceptance is a special case of so called promise problem solving.
A {\em promise problem} is a pair $A = (A_{yes}, A_{no})$, where $A_{yes}$, $A_{no}\subset \Sigma^*$
are disjoint sets. Languages may be viewed as promise problems that obey the additional constraint
$A_{yes}\cup A_{no}=\Sigma^*$.

A 1DFA ${\cal A}$ solves a promise problem $A = (A_{yes}, A_{no})$ (exactly) if
\begin{enumerate}
\item[1.] $\forall w\in A_{yes}$, ${\cal A}$ accepts $w$, and
\item[2.] $\forall w\in  A_{no}$, ${\cal A}$ rejects $w$.
\end{enumerate}

Let $0<\varepsilon\leq\frac{1}{2}$. An MO-1QCFA ${\cal A}$ solves a promise problem $A = (A_{yes}, A_{no})$ with a one-sided error $\varepsilon$ if
\begin{enumerate}
\item[1.] $\forall w\in A_{yes}$, $Pr[{\cal A}\  \text{accepts}\  w]=1$, and
\item[2.] $\forall w\in  A_{no}$, $Pr[{\cal A}\ \text{rejects}\  w]\geq 1-\varepsilon$.
\end{enumerate}

 A promise problem $A = (A_{yes}, A_{no})$ is solved  by exactly by an  MO-1QCFA ${\cal A}$  if
\begin{enumerate}
\item[1.] $\forall w\in A_{yes}$, $Pr[{\cal A}\  \text{accepts}\  w]=1$, and
\item[2.] $\forall w\in  A_{no}$, $Pr[{\cal A}\ \text{rejects}\  w]=1$.
\end{enumerate}

\section{Communication complexity for a promise version of the disjointness problem}
In this section we give  quantum and probabilistic upper bounds and a deterministic low bound for $\text{DISJ}_{\lambda}'$.
\subsection{Quantum protocol}
We give at first a quantum communication protocol for $\text{DISJ}_{\frac{1}{4}}'(x,y)$.

\begin{Th}\label{th1}
The quantum communication complexity $Q_2(DISJ_{\frac{1}{4}}')\leq 3+2\log n$.
\end{Th}
\begin{proof}
Assume that Alice is given an input $x=x_1,\cdots,x_n$ and Bob an input $y=y_1,\cdots,y_n$. The quantum communication protocol ${\cal P}$ works as follows:
\begin{enumerate}
  \item Alice starts with a quantum state $|\psi_0\rangle=|1,0\rangle=(1, \overbrace{0,\cdots,0}^{2n-1})^T$ and applies the following unitary transformation $U_s$:
  \begin{equation}
    U_s|\psi_0\rangle=\sum_{i=1}^n\frac{1}{\sqrt{n}}|i,0\rangle=\frac{1}{\sqrt{n}}(\overbrace{1,\cdots,1}^{n}, \overbrace{0,\cdots,0}^{n})^T.
  \end{equation}
  Alice then applies  unitary transformation $U_x$ according to her input $x=x_1,\cdots,x_n$:
  \begin{equation}
    U_x=U_{x_n}\cdots U_{x_1}
  \end{equation}
  where
   \begin{equation}
    U_{x_i}=\left\{\begin{array}{ll}
                    I &\ \text{if}\  x_i=0 \\
                    |i,1\rangle\langle i,0|+|i,0\rangle\langle i,1| +\sum_{j\neq i}|j,0\rangle\langle j,0|+\sum_{j\neq i}|j,1\rangle\langle j,1|&\ \text{if}\ x_i=1\\
                  \end{array}
 \right..
\end{equation}

$U_x$ can be seen as a unitary transformation that exchanges the amplitudes of $|i,0\rangle$ and $|i,1\rangle$ if $x_i=1$.
  The quantum state after performing $U_x$ is
    \begin{equation}
    |\psi_1\rangle=\frac{1}{\sqrt{n}}(\overline{x}_1,\cdots,\overline{x}_n,x_1,\cdots,x_n)^T,
  \end{equation}
    where $\overline{x}_i=1-x_i$.
  Alice then sends her quantum state $|\psi_1\rangle$ to Bob.
  \item Bob then applies to the state received a  unitary mapping $V_y$, defined for his input $y$ as follows
   \begin{equation}
    V_y|i,0\rangle=|i,0\rangle
  \end{equation}
  and
   \begin{equation}
    V_y|i,1\rangle=(-1)^{y_i}|i,1\rangle.
  \end{equation}
  The quantum state after applying $V_y$ is
   \begin{align}
    |\psi_2\rangle=\frac{1}{\sqrt{n}}(\overline{x}_1,\cdots,\overline{x}_n,(-1)^{y_1}x_1,\cdots,(-1)^{y_n}x_n)^T.
  \end{align}
  If $x_i=y_i=1$, then $(-1)^{y_i}x_i=-1=(-1)^{x_i\wedge y_i}$; if $x_i=1$ and $y_i=0$, then $(-1)^{y_i}x_i=1=(-1)^{x_i\wedge y_i}$; otherwise  $(-1)^{y_i}x_i=0$.
  \item Alice applies the unitary transformation $U_x$ to the quantum state $|\psi_2\rangle$ received from Bob and gets a new quantum state:
    \begin{equation}
    |\psi_{3}\rangle=\frac{1}{\sqrt{n}}( z_1,\cdots,z_n,\overbrace{0,\cdots,0}^{n})^T.
  \end{equation}
  If $x_i=0$, then $z_i=\overline{x}_i=1=(-1)^{x_i\wedge y_i}$. If $x_i=1$, then $z_i=(-1)^{y_i}x_i=(-1)^{x_i\wedge y_i}$.
  Therefore, $z_i=(-1)^{x_i\wedge y_i}$ for $1\leq i\leq n$.

  Alice then applies the unitary transformation $U_f$ ($U_f$  which will be specified later)  to get the follow state:
   \begin{equation}
    U_f |\psi_{3}\rangle=(  \frac{1}{n}\sum_{i=1}^n (-1)^{x_i\wedge y_i}, \overbrace{ *,\cdots, *}^{2n-1}  )^T.
  \end{equation}
and then measures the final quantum state with $\{|i,0\rangle\langle i,0|,|i,1\rangle\langle i,1| \}_{i=1}^n$. If the measurement outcome is $|1,0\rangle$, then Alice sends 1 to Bob; otherwise, Alice sends 0 to Bob.
\end{enumerate}

It is clear that this protocol communicates $1+2(\log{2n})=3+2\log n$ qubits.   Unitary transformations $U_{s}$ and $U_{f}$ do exist. The first column of $U_{s}$ is $\frac{1}{\sqrt{n}}(\overbrace{1,\cdots,1}^{n},\overbrace{0,\cdots,0}^n)^T$ and the first row of $U_f$ is $\frac{1}{\sqrt{n}}(\overbrace{1,\cdots,1}^{n},\overbrace{0,\cdots,0}^n)$.
It is easy to verify that $V_y$ is a unitary transformation.

If $\sum_{i=1}^n x_i\wedge y_i=0$, then $ \frac{1}{n}\sum_{i=1}^n (-1)^{x_i\wedge y_i}=1$. After the measurement, Alice gets the quantum outcome $|1,0\rangle$ and sends 1 to Bob with certainty. Thus,
 \begin{equation}
    Pr({\cal P}(x,y)=\text{DISJ}_{\frac{1}{4}}'(x,y))=1.
  \end{equation}

If $n/4\leq \sum_{i=1}^n x_i\wedge y_i\leq 3n/4$, then $|\frac{1}{n}\sum_{i=1}^n (-1)^{x_i\wedge y_i}|\leq 1/2$.  Alice gets the quantum outcome
$|1,0\rangle$ with probability not more than $|\frac{1}{n}\sum_{i=1}^n (-1)^{x_i\wedge y_i}|^2=1/4$. Thus,
 \begin{equation}
    Pr({\cal P}(x,y)=\text{DISJ}_{\frac{1}{4}}'(x,y))=1-|\frac{1}{n}\sum_{i=1}^n (-1)^{x_i\wedge y_i}|^2 \geq \frac{3}{4}.
  \end{equation}
  Therefore ${\cal P}$ is a two-sided error quantum protocol for $\text{DISJ}_{\frac{1}{4}}'(x,y))$ and $Q_2(DISJ_{\frac{1}{4}}')\leq 3+2\log n$.
\end{proof}

\begin{Th}\label{q-g}
The  quantum communication complexity $Q_2(DISJ_{\lambda}')\leq \frac{\log 3}{3\lambda}(3+2\log n)$.
\end{Th}

\begin{proof}
For general cases,  the new quantum protocol  ${\cal P}'$ works as follows:
Repeat the protocol ${\cal P}$ from the proof of the previous theorem  $k$ times ($k$ will be specified later). If all the measurement outcomes in Step 3 are $|1,0\rangle$, then ${\cal P}'(x,y)=1$; otherwise, ${\cal P}'(x,y)=0$.

If $\sum_{i=1}^n x_i\wedge y_i=0$, then
\begin{equation}
    Pr({\cal P}(x,y)=1)=1
\end{equation}
 and
 \begin{equation}
    Pr({\cal P}(x,y)=0)=0.
\end{equation}
Therefore,
 \begin{align}
    Pr({\cal P}'(x,y)&=\text{DISJ}_{\lambda}'(x,y)=1)=1.
\end{align}

If $\lambda n\leq \sum_{i=1}^n x_i\wedge y_i\leq (1-\lambda)n$, then
  \begin{align}
    p_0&=Pr({\cal P}(x,y)=\text{DISJ}_{\lambda}'(x,y)=0)=1-|\frac{1}{n}\sum_{i=1}^n (-1)^{x_i\wedge y_i}|^2 \geq 1-|1-2\lambda|^2\\
    &=4\lambda-\lambda^2=4\lambda(1-\lambda)\geq 4\lambda(1-\frac{1}{4})=3\lambda.
\end{align}
 If $k=\frac{\log 1/3}{\log (1-3\lambda)}$, and the protocol is repeated ${\cal P}$ $k$ times, then
 \begin{align}
    Pr({\cal P}'(x,y)&=\text{DISJ}_{\lambda}'(x,y)=0)=1-(1-p_0)^k\geq 1-(1-3\lambda)^k\geq 1-(1-3\lambda)^{\frac{\log 1/3}{\log (1-3\lambda)}}\\
    &=1-2^{\log ((1-3\lambda)^{\frac{\log 1/3}{\log (1-3\lambda)}})}=1-2^{\frac{\log 1/3}{\log (1-3\lambda)}\times \log ((1-3\lambda)}=1-2^{\log{1/3}}=\frac{2}{3}
    .
\end{align}
Since for any real number $u>0$, $1-u\leq e^{-u}\leq 2^{-u}$.  We have,
 \begin{align}
  k=\frac{\log 1/3}{\log (1-3\lambda)}\leq \frac{\log 1/3}{\log 2^{(-3\lambda)}}=\frac{\log 3}{3\lambda}.
\end{align}
Thus, $Q_2(\text{DISJ}_{\lambda}')\leq \frac{\log 3}{3\lambda}(3+2\log n)$.

\end{proof}

\subsection{Deterministic lower bound}

\begin{Th}\label{th3}
The  deterministic communication complexity $D(DISJ_{\lambda}')\in {\bf \Omega}(n)$.
\end{Th}
\begin{proof}
Let ${\cal P}$ be a deterministic protocol for $\text{DISJ}_{\lambda}'$. We consider the set $F_{\lambda}=\{x\in\{0,1\}^n \,|\, \lambda n\leq   H(x)\leq (1-\lambda)n\}$, where $H(x)$ is the Hamming weight of $x$.
If $x\in F_{\lambda}$, then also $\overline{x}\in F_{\lambda}$, where $\overline{x}=\overline{x}_1\cdots \overline{x}_n$.
If the set $E=\{(x,\overline{x}) \,|\, x\in F_{\lambda}\}$, then for every $(x,\overline{x})\in E$, ${\cal P}(x,\overline{x})=1$.
Suppose there is a 1-monochromatic rectangle $R=A\times B\subseteq \{0,1\}^n\times\{0,1\}^n$ such that ${\cal P}(x,y)=1$ for every $(x,y)\in R$. For $S=R\cap E$,  we now prove that $|S|<1.99^n$.

Suppose $|S|\geq 1.99^n$. According to Corollary 1.2 from \cite{Fr87},  then there exist $(x,\overline{x})\in S$ and $(z,\overline{z})\in S$ such that $x\wedge z=\frac{n}{4}$, where $x\wedge z=\sum_{i=1}^n x_i\wedge z_i$. Since $S\subseteq E$, we have $x, \overline{x},z,\overline{z}\in F_{\lambda}$. Without a lost of generality, let
\begin{align}
    x&=\overbrace{1\cdots 1}^{n/4}\ \overbrace{ 0\cdots 0}^{\lambda n}\  \overbrace{ 1\cdots 1}^{\lambda n}\ \overbrace{ *\cdots *}^{3n/4-2\lambda n},\text{and}\\
    z&=\overbrace{1\cdots 1}^{n/4}\ \overbrace{ 1\cdots 1}^{\lambda n}\  \overbrace{ 0\cdots 0}^{\lambda n}\ \overbrace{ *\cdots *}^{3n/4-2\lambda n}.
\end{align}
In such a case
\begin{align}
    \overline{x}=\overbrace{0\cdots 0}^{n/4}\ \overbrace{ 1\cdots 1}^{\lambda n}\  \overbrace{ 0\cdots 0}^{\lambda n}\ \overbrace{ *\cdots *}^{3n/4-2\lambda n}
\end{align}
and therefore $\lambda n \leq z\wedge \overline{x}\leq 3n/4-\lambda n <(1-\lambda) n$. Thus, ${\cal P}(z,\overline{x})=0$. Since $S\subset R$ and $R$ is 1-rectangle, we get $(x,\overline{x})\in R, (z,\overline{z})\in R$ and therefore also $(z,\overline{x})\in R$~--~ a contradiction.

Therefore, the minimum number of 1-monochromatic rectangles that partition the space of inputs is
 \begin{align}
     C^1(\text{DISJ}_{\lambda}')\geq \frac{|E|}{|S|}= \frac{|F_{\lambda}|}{|S|}\geq \frac{|F_{1/4}|}{|S|}\geq \frac{2^n/2}{1.99^n}.
\end{align}
According to Lemma \ref{lm-d-lowbound}, the deterministic communication complexity is then
 \begin{align}
D(\text{DISJ}_{\lambda}')\geq \log{C^1(\text{DISJ}_{\lambda}')}\geq \log{(\frac{2^n/2}{1.99^n})}=n-1-n\log{1.99}>n-1-0.9927n=0.0073n-1.
 \end{align}
 Thus, $D(\text{DISJ}_{\lambda}')\in {\bf \Omega}(n)$.
\end{proof}

\begin{Rm}
The lower bound that was proved in the previous theorem is quite a rough bound. We expect that the low bound is relative to $\lambda$. When $\lambda$ is close to 0, then the low bound is expected to be close to $n$ instead of 0.007n.
\end{Rm}

\subsection{Probabilistic protocol}
As already mentioned,  the two-sided error  probabilistic communication complexity  $R_2(\text{DISJ}) \in{\bf \Omega}(n)$.  However, for $\text{DISJ}'$, the  communication complexity is dramatically improved as will now be shown

\begin{Th}
The one-way probabilistic communication complexity $R_2^1(DISJ_{\frac{1}{4}}')\leq 5\log{n}$.
\end{Th}
\begin{proof}
Let us consider the one-way probabilistic protocol ${\cal P}^1$ which works as follows:
\begin{enumerate}
  \item If $H(x)<k$, then Alice just sends 1 to Bob.  Otherwise, Alice chooses randomly $k$ bits of `1's of her input, says $x_{i_1},\cdots, x_{i_k}$,  and sends the positions $i_1,\cdots, i_k$ to Bob.
  \item If Bob just receives a bit of `1', then ${\cal P}^1(x,y)=1$. Otherwise, Bob checks the positions $i_1,\cdots, i_k$ of his input. If there exists a $j$ $(1\leq j\leq k)$ such that $y_{i_j}=1$, then ${\cal P}^1(x,y)=0$; otherwise ${\cal P}^1(x,y)=1$.
\end{enumerate}

If $\sum_{i=1}^n x_i\wedge y_i=0$, then
\begin{equation}
    Pr({\cal P}^1(x,y)=\text{DISJ}_{\frac{1}{4}}'(x,y)=1)=1.
\end{equation}

If $n/4\leq \sum_{i=1}^n x_i\wedge y_i\leq 3n/4$, then for any $i\in\{i_1,\cdots,i_k\}$
\begin{equation}
    Pr(y_{i}=x_i)\geq \frac{1}{4}.
\end{equation}
Therefore,
\begin{equation}
   Pr({\cal P}^1(x,y)=0)\geq 1-(1-{\frac{1}{4}})^k=1-(\frac{3}{4})^k.
\end{equation}
If $k=5$, then $Pr({\cal P}^1(x,y)=0)\geq \frac{2}{3}$. Since Alice needs $\log{n}$ bits to specifies every position, we have $R_2^1(\text{DISJ}_{\frac{1}{4}}')\leq 5\log{n}$.
\end{proof}

\begin{Th}\label{p-g}
The one-way probabilistic communication complexity $R_2^1(DISJ_{\lambda}')\leq \frac{\log{3}}{\lambda}\log{n}$.
\end{Th}
\begin{proof}
For this general cases, we use the same protocol as in the proof of the previous theorem, but Alice will  send more positions of her `1' bits input.

If $\sum_{i=1}^n x_i\wedge y_i=0$, then
\begin{equation}
    Pr({\cal P}^1(x,y)=\text{DISJ}_{\frac{1}{4}}'(x,y)=1)=1.
\end{equation}
If $\lambda n\leq \sum_{i=1}^n x_i\wedge y_i\leq (1-\lambda)n$,  then for any $i\in\{i_1,\cdots,i_k\}$
\begin{equation}
    Pr(y_{i}=x_i)\geq \lambda.
\end{equation}
Therefore,
\begin{equation}
   Pr({\cal P}^1(x,y)=0)\geq 1-(1-\lambda)^k.
\end{equation}
If $k=\frac{\log{1/3}}{\log{(1-\lambda)}}$, then $Pr({\cal P}^1(x,y)=0)\geq \frac{2}{3}$. Thus, $R_2^1(\text{DISJ}_{\lambda}')\leq \frac{\log{1/3}}{\log{(1-\lambda)}}\log{n}\leq \frac{\log{3}}{\lambda}\log{n}$.
\end{proof}

\begin{Rm}
If $\lambda$ is close to  $\frac{1}{n}$, then the quantum and probabilistic communication complexity advantages in Theorem \ref{q-g} and Theorem \ref{p-g} disappear.  We can define two-sided error mode as tolerating an error  probability   $\varepsilon$ instead of $\frac{1}{3}$. Modifying our proof in  Theorem \ref{q-g} and Theorem \ref{p-g}, we can get  $Q_2(\text{DISJ}_{\lambda}')\leq \frac{\log \varepsilon}{3\lambda}(3+2\log n)$ and $R_2^1(\text{DISJ}_{\lambda}')\leq \frac{\log{\varepsilon}}{\lambda}\log{n}$ for any error  probability $\varepsilon$. When $\varepsilon$ is close to $\frac{1}{2^n}$, the  the quantum and probabilistic communication complexity advantages in Theorem \ref{q-g} and Theorem \ref{p-g} disappear.
\end{Rm}

\section{Applications to finite automata}

For any $n\in {\mathbb{Z}}^+$,  we consider two promise problems $A_{EQ}(n)$ and $A_{D}(n)$ corresponding to $\text{EQ}'$ and $\text{DISJ}_{\frac{1}{4}}'$ that are defined as follows:
 \begin{align}
&A_{EQ}(n):\left\{\begin{array}{l}
                    A_{yes}(n)=\{x\#y\,|\,x=y,x,y\in\{0,1\}^n\} \\
                    A_{no}(n)=\{x\#y\,|\,x\neq y,x,y\in\{0,1\}^n, H(x,y)= \frac{n}{2}\}
                  \end{array}
 \right.,\\
&A_{D}(n):\left\{\begin{array}{l}
                    A_{yes}(n)=\{x\#y\#x\,|\,\sum_{i=1}^n x_i\wedge y_i=0,x,y\in\{0,1\}^n\} \\
                    A_{no}(n)=\{x\#y\#x\,|\,\frac{1}{4} n\leq \sum_{i=1}^n x_i\wedge y_i\leq \frac{3}{4} n,x,y\in\{0,1\}^n\}
                  \end{array}
 \right..
\end{align}

The quantum protocol for $\text{EQ}'$ which is described  in \cite{Buh09} can be implemented on an MO-1QCFA as shown bellow.
Therefore, we get the following result:

\begin{Th}
The promise problem $A_{EQ}(n)$ can be solved exactly by an MO-1QCFA  ${\cal A}(n)$ with $n$ quantum basis states and ${\bf O}(n)$ classical states, whereas the sizes of the corresponding 1DFA  are $2^{{\bf \Omega}(n)}$.
\end{Th}

\begin{figure}[htbp]
 %  %Requires \usepackage{graphicx}
\begin{tabular}{|l|}
    \hline

\begin{minipage}[t]{0.93\textwidth}
\begin{enumerate}
\item[1.] Read the left end-marker $\ |\hspace{-1.5mm}c$,  perform $U_s$ on the initial quantum state $|1\rangle$,  change its classical state to $\delta(s_0,\ |\hspace{-1.5mm}c )=s_1$, and move the tape head one cell to the right.

\item[2.] Until the currently  scanned symbol $\sigma$ is not $\#$, do the following:
 \begin{enumerate}
 \item[2.1] Apply $\Theta(s_i,\sigma)=U_{i,\sigma}$ to the current quantum state.
 \item[2.2] Change the classical state $s_i$ to $s_{i+1}$ and move the tape head one cell to the right.
\end{enumerate}
\item[3.] Change the classical state $s_{n+1}$ to  $s_1$ and move the tape head one cell to the right.

\item[4.] While the currently  scanned symbol $\sigma$ is not the right end-marker $\$$, do the following:
 \begin{enumerate}
 \item[2.1] Apply $\Theta(s_i,\sigma)=U_{i,\sigma}$ to the current quantum state.
 \item[2.2] Change the classical state $s_i$ to $s_{i+1}$ and move the tape head one cell to the right.
\end{enumerate}

\item[5.] When the right end-marker  is reached,    perform $U_{f}$ on the current quantum state,
measure the current quantum state with $M=\{P_i=|i\rangle\langle i|\}_{i=1}^{n}$.   If the outcome is $|1\rangle$, accept the input; otherwise reject the input.

\end{enumerate}

\end{minipage}\\

\hline
\end{tabular}
 \centering\caption{  Description of the behavior of ${\cal A}(n)$ when solving the promise problem $A_{EQ}(n)$. }\label{f3}
\end{figure}

\begin{proof}
Let $x=x_1\cdots x_n$ and $y=y_1\cdots y_n$.  Let us consider an MO-1QCFA ${\cal A}(n)$ with $n$   quantum basis states
$\{|i\rangle:i=1,2,\cdots,n\}$. ${\cal A}(n)$ will start in  the
quantum state $|1\rangle=(1,0\cdots,0)^T$. We use classical states $s_i\in S$ ($1\leq i\leq n+1$) to point out the positions of the tape head that will provide some information for  quantum transformations. If  the classical state of ${\cal A}(n)$ will be $s_i$ ($1\leq i\leq n$) that will  mean that the next scanned symbol of the tape head is the $i$-th symbol of $x$($y$) and $s_{n+1}$ means that the next scanned symbol of  the tape head  is  $\#$($\$$).
 The automaton proceeds as shown in Figure \ref{f3}, where
\begin{align}
&U_s|1\rangle=\frac{1}{\sqrt{n}}\sum_{i=1}^n|i\rangle;\\
&U_{i,\sigma}|i\rangle=(-1)^{\sigma}|i\rangle \text{ and }  U_{i,\sigma}|j\rangle=|j\rangle \ \text{for}\ j\neq i\\
&U_f(\sum_{i=1}^n\alpha_i|i\rangle)=(\frac{1}{\sqrt{n}}\sum_{i=1}^n\alpha_i)|1\rangle+\cdots
\end{align}

 Transformations $U_{s}$ and $U_{f}$ are unitary. The first column of $U_{s}$ is $\frac{1}{\sqrt{n}}(1,\cdots,1)^T$ and the first row of $U_f$ is $\frac{1}{\sqrt{n}}(1,\cdots,1)$.

The quantum state after scanning the left end-marker is $|\psi_1\rangle= U_s|1\rangle=\sum_{i=1}^n\frac{1}{\sqrt{n}}|i\rangle$, the quantum state after Step 2 is $|\psi_2\rangle=\sum_{i=1}^n\frac{1}{\sqrt{n}}(-1)^{x_i}|i\rangle$, and the quantum state after Step 4 is $|\psi_{3}\rangle=\sum_{i=1}^n\frac{1}{\sqrt{n}}(-1)^{x_i+y_i}|i\rangle$. The quantum state after scanning  the right end-marker  is therefore
\begin{align}
|\psi_4\rangle=U_f\left(\sum_{i=1}^n\frac{1}{\sqrt{n}}(-1)^{x_i+y_i}|i\rangle\right)=U_f\frac{1}{\sqrt{n}} \left(
                                                  \begin{array}{c}
                                                    (-1)^{x_1+y_1} \\
                                                    (-1)^{x_2+y_2} \\
                                                    \vdots \\
                                                    (-1)^{x_n+y_n}\\
                                                  \end{array}
                                                \right) =\left(
                                                  \begin{array}{c}
                                                    \frac{1}{n} \sum_{i=1}^n (-1)^{x_i+y_i} \\
                                                    \vdots \\
                                                    \vdots \\
                                                  \end{array}
                                                \right).
\end{align}
 % \sum_{i=1}^n(-1)^{x_i+y_i}\sum_{j=1}^n(-1)^{i\cdot j}|j\rangle.

If the input string $w\in A_{yes}(n)$, then $x_i=y_i$ for $1\leq i\leq n$ and  $|\frac{1}{n} \sum_{i=1}^n (-1)^{x_i+y_i}|^2=1$. The amplitude of $|1\rangle$ is 1, and that means $|\psi_4\rangle=|1\rangle$.
Therefore the input will be accepted with probability 1 at the measurement in Step 5.

If the input string $w\in A_{no}(n)$, then $H(x,y)=\frac{n}{2}$. Therefore   the probability of getting outcome $|1\rangle$ in the measurement in Step 5 is $|\frac{1}{n} \sum_{i=1}^n (-1)^{x_i+y_i}|^2=0$.

The deterministic communication complexity for  $EQ'$is at least $0.007n$ \cite{Buh98,Buh09}. Therefore,  the sizes of the corresponding 1DFA are $2^{{\bf \Omega}(n)}$.
\end{proof}

We implement the quantum protocol used in Theorem \ref{th1} on an MO-1QCFA and prove the following theorem:
\begin{Th}
The promise problem $A_{D}(n)$ can be solved with one-sided error $\frac{1}{4}$ by an MO-1QCFA  ${\cal A}(n)$ with $2n$ quantum basis states and ${\bf O}(n)$ classical states, whereas the sizes of the corresponding 1DFA  are $2^{{\bf \Omega}(n)}$.
\end{Th}

\begin{figure}[htbp]
 %  %Requires \usepackage{graphicx}
\begin{tabular}{|l|}
    \hline
\begin{minipage}[t]{0.93\textwidth}
\begin{enumerate}
\item[1.] Read the left end-marker $\ |\hspace{-1.5mm}c$,  perform $U_s$ on the initial quantum state $|1,0\rangle$,  change its classical state to $\delta(s_0,\ |\hspace{-1.5mm}c )=s_1$, and move the tape head one cell to the right.

\item[2.] While the currently  scanned symbol $\sigma$ is not $\#$, do the following:
 \begin{enumerate}
 \item[2.1] Apply $\Theta(s_i,\sigma)=U_{i,\sigma}$ to the current quantum state.
 \item[2.2] Change the classical state $s_i$ to $s_{i+1}$ and move the tape head one cell to the right.
\end{enumerate}
\item[3.] Move the tape head one cell to the right.

\item[4.] While the currently  scanned symbol $\sigma$ is not $\#$, do the following:
 \begin{enumerate}
 \item[4.1] Apply $\Theta(s_{n+i},\sigma)=V_{i,\sigma}$ to the current quantum state.
 \item[4.2] Change the classical state $s_{n+i}$ to $s_{n+i+1}$ and move the tape head one cell to the right.
\end{enumerate}

\item[5.] Change the classical state $s_{2n+1}$ to  $s_1$ and move the tape head one cell to the right.

\item[6.] While the currently  scanned symbol $\sigma$ is not the right end-marker $\$$, do the following:
 \begin{enumerate}
 \item[6.1] Apply $\Theta(s_i,\sigma)=U_{i,\sigma}$ to the current quantum state.
 \item[6.2] Change the classical state $s_i$ to $s_{i+1}$ and move the tape head one cell to the right.
\end{enumerate}

\item[7.] When the right end-marker $\$$ is reached,    perform $U_{f}$ on the current quantum state,
measure the current quantum state with $M=\{|i,0\rangle\langle i,0|,|i,1\rangle\langle i,1| \}_{i=1}^n$.   If the outcome is $|1,0\rangle$, accept the input; otherwise reject the input.

\end{enumerate}

\end{minipage}\\

\hline
\end{tabular}
 \centering\caption{  Description of the behavior of ${\cal A}(n)$ when solving the promise problem $A_{D}(n)$. }\label{f3}
\end{figure}

\begin{proof}

Let $x=x_1\cdots x_n$ and $y=y_1\cdots y_n$.  Let us consider an MO-1QCFA ${\cal A}(n)$ with $2n$   quantum basis states
$\{|i,0\rangle,|i,1\rangle \}_{i=1}^n$ that will start in  the
 state $|1,0\rangle=(1, \overbrace{0,\cdots,0}^{2n-1})^T$.
 The automaton proceeds as shown in Figure \ref{f3}, where
\begin{align}
&U_s|1,0\rangle=\frac{1}{\sqrt{n}}\sum_{i=1}^n|i,0\rangle;\\
& U_{i,\sigma}|j,0\rangle=|j,1\rangle \ \text{and}\ U_{i,\sigma}|j,1\rangle=|j,0\rangle \ \text{if } \sigma=1 \ \text{and}\ j=i \ \text{otherwise}\  U_{i,\sigma}|j,k\rangle=|j,k\rangle;\\
& V_{i,\sigma}|j,1\rangle=(-1)^{\sigma}|j,1\rangle \ \text{if } j=i ,\  \text{otherwise}\  V_{i,\sigma}|j,k\rangle=|j,k\rangle;\\
&U_f(\sum_{i=1}^n\alpha_i|i,0\rangle+\beta_i|i,1\rangle)=(\frac{1}{\sqrt{n}}\sum_{i=1}^n\alpha_i)|1,0\rangle+\cdots.
\end{align}

For $1\leq i\leq n$, it is easy to verify that $U_s, U_{i,\sigma}, V_{i,\sigma}$ and $U_f$ are unitary transformations. According to the analysis in the proof of Theorem \ref{th1},
if the input string $w\in A_{yes}(n)$, the automaton will get the outcome $|1,0\rangle$ in Step 7 with certainty and therefore
\begin{align}
Pr[{\cal A}\  \text{accepts}\  w]=1.
\end{align}
If the input string $w\in A_{no}(n)$,  the automaton gets the  outcome
$|1,0\rangle$ with probability not more than $1/4$. Thus,
\begin{align}
Pr[{\cal A}\  \text{rejects}\  w]\geq \frac{3}{4}.
\end{align}

Let  an $N$-states
1DFA ${\cal A}'(n)=(S,\Sigma,\delta,s_0,S_{acc})$ solves the promise problem $A_{D}(n)$,  then we can get a deterministic protocol for $\text{DISJ}_{\frac{1}{4}}'(x,y)$ as follows:
\begin{enumerate}
  \item Alice simulates the computation of ${\cal A}'(n)$ on input ``$x\#$" and then sends her state $\widehat{\delta}(s_0,x\#)$ to Bob.
  \item Bob simulates the computation of ${\cal A}'(n)$ on input ``$y\#$" starting at state $\widehat{\delta}(s_0,x\#)$, and then sends his state  $\widehat{\delta}(s_0,x\#y\#)$ to Alice.
  \item Alice simulates the computation of ${\cal A}'(n)$ on input ``$x$" starting at state $\widehat{\delta}(s_0,x\#y\#)$. If $\widehat{\delta}(s_0,x\#y\#x)\in S_acc$ then Alice sends 1 to Bob, otherwise Alice sends 0 to Bob.
\end{enumerate}

The deterministic complexity of the above protocol is $ 1+2 \log{N}$ and therefore $D(\text{DISJ}_{\frac{1}{4}}')\leq 1+2 \log{N}$. According the analysis in  \ref{th3}, we have
\begin{align}
&1+2 \log{N}\geq D(\text{DISJ}_{\frac{1}{4}}')> 0.0073n-1\\
&\Rightarrow N\in 2^{{\bf \Omega}(n)}.
\end{align}

\end{proof}

\section*{Acknowledgements}

Work of the first and third authors supported by
the Employment of Newly Graduated Doctors of Science for Scientific Excellence project/grant (CZ.1.07./2.3.00\linebreak[0]/30.0009)  of Czech Republic.
Work of second author   supported by the National
Natural Science Foundation of China (Nos. 61272058, 61073054).

\end{document}